\documentclass[12pt,reqno]{amsart}
\usepackage[margin=1in]{geometry}
\usepackage{hyperref,amssymb,amsmath,enumerate,amsthm}
\usepackage[all]{xy}
\usepackage{mathtools}

\usepackage{graphicx}
\usepackage{epstopdf}
\usepackage[table]{xcolor}

\usepackage{indentfirst}
\usepackage{latexsym}
\usepackage{caption}

\usepackage{enumitem}

\usepackage{tikz}
\usetikzlibrary{intersections,backgrounds}

\usepackage{minted}
\usepackage{qcircuit}

\newtheorem{thm}{Theorem}[section]
\newtheorem{lemma}[thm]{Lemma}

\newtheorem{defn}[thm]{Definition}
\newtheorem{remark}[thm]{Remark}

\usepackage{array}
\newcommand{\PreserveBackslash}[1]{\let\temp=\\#1\let\\=\temp}
\newcolumntype{C}[1]{>{\PreserveBackslash\centering}p{#1}}
\newcolumntype{R}[1]{>{\PreserveBackslash\raggedleft}p{#1}}
\newcolumntype{L}[1]{>{\PreserveBackslash\raggedright}p{#1}}

\newcommand{\A}{\mathcal{A}}
\newcommand{\D}{\Delta}

\newcommand{\C}{\mathcal{C}}
\newcommand{\F}{\mathbb{F}}

\renewcommand{\O}{\mathcal{O}}

\newcommand{\G}{\mathcal{G}}

\newcommand{\ket}[1]{\ensuremath{\left|{#1}\right\rangle}}
\usepackage{graphicx,subcaption}
\usepackage[font=small,labelfont=bf]{caption}

\usepackage{mathtools}
\usepackage{textcase}
\makeatletter
\let\uppercasenonmath\@firstofone 
\makeatother

\usepackage{amsaddr} 

\title{Isogeny Graphs in Superposition and \mbox{Quantum Onion Routing}}

\author[Eleni Agathocleous, Tobias Hartung, Karl Jansen, Lukas Mansour]{\mbox{Eleni Agathocleous$^1$, Tobias Hartung$^{2,3}$, Karl Jansen$^{1,4}$, Lukas Mansour$^{4}$}}\normalsize
\address{$^1$ Computation-Based Science and Technology Research Center, The Cyprus Institute\\
20 Kavafi Street, 2121 Nicosia, Cyprus \\
$^2$ Northeastern University -- London, Devon House, St Katharine Docks, London, E1W 1LP, United Kingdom \\
$^3$ Khoury College of Computer Sciences, Northeastern University, \#202, West Village Residence Complex H, 440 Huntington Ave, Boston, 02115, MA, USA\\
$^4$ Deutsches Elektronen-Synchrotron DESY, Platanenallee 6, 15738 Zeuthen, Germany}

\begin{document}

\begin{abstract}
\begin{center}\end{center}
Onion routing provides anonymity by layering encryption so that no relay can link sender to destination. A quantum analogue faces a core obstacle: layered quantum encryption generally requires symmetric encryption schemes whereas classically one would heavily rely on public key encryption. We propose a symmetric-encryption-based quantum onion routing (QOR) scheme by instantiating each layer with the abelian ideal class group action from the Theory of Complex Multiplication. Session keys are established locally via a Diffie-Hellman key exchange between neighbors in the chain of communication. Furthermore, we propose a novel ``non-local'' key exchange between the sender and receiver. The underlying problem remains hard even for quantum adversaries and underpins the security of current post-quantum schemes. We connect our construction to isogeny graphs and their association schemes, using the Bose--Mesner algebra to formalize commutativity and guide implementation. We give two implementation paths: (i) a universal quantum oracle evaluating the class group action with polynomially many quantum resources, and (ii) an intrinsically quantum approach via continuous-time quantum walks (CTQWs), outlined here and developed in a companion paper. A small Qiskit example illustrates the mechanics (by design, not the efficiency) of the QOR. 
\end{abstract}

\maketitle

\section{Introduction} Onion routing is a privacy-enhancing network protocol whose primary goal is anonymity. Beyond transport security, it uses layered encryption so no single relay can link sender to destination. Onion routing anonymises network traffic by wrapping it in multiple layers of encryption and sending it through a short chain of relays (guard $\rightarrow$ middle $\rightarrow$ exit). Each relay peels one layer and knows only the next hop, so no single node can link source and destination. If sufficient network traffic is present, this ensures that an outside observer can only link the sender of a message to the guard server and the receiver to the exit server, but tracking the communication through the network is difficult. Thus, onion routing not only keeps the message hidden but also hides who is communicating with whom. 

A quantum analogue of onion-routing encounters challenges intrinsic to quantum computing. First and foremost, the fact that every quantum operation is unitary implies that any operation performed on a state is reversible with relative ease. Thus, any cryptographic setup in which the message is encoded as a state and the encryption works as a function (determined by the key) on the message state is necessarily a symmetric encryption scheme. Public key encryption in a quantum computing setting therefore (up to further complications) has the key as a state and the message defining the operation performed, or has a fixed state and both key and message define the operation. The decryption process then relies on identifying which operation was performed by the sender and reconstructing the message from that information. Although theoretically possible, such an asymmetric encryption scheme cannot easily be layered, if the message is a genuine quantum state (as opposed to a bit-state encoded in qubits).

The main goal of this paper is to overcome these problems, and design a quantum onion routing through layered symmetric quantum encryption. We achieve this by basing each layer on the \emph{ideal class group action} from Number Theory, and specifically from the \emph{Theory of Complex Multiplication}. The associated group-action inversion (vectorization) problem remains hard even for quantum adversaries, and underpins post-quantum schemes such as CSIDH \cite{CSIDH}. ``Local session keys'' are generated via Diffie-Hellman key exchange for neighboring nodes in the chain of communication. Furthermore, we propose a method of communicating a ``global session key'' chosen by the sender to the receiver that allows for the final decryption step. 

Of course, many of the individual components of the proposed quantum onion routing scheme can be replaced with quantum encryption schemes that serve a similar role. E.g., local session keys can be established using quantum key distribution and the global session key encryption layer can be replaced with a quantum public key encryption layer. However, it is advantageous to base the entire scheme on a single underlying problem, such as for ease of implementation. This means that all encryption layers need to use compatible schemes. The isogeny-based post-quantum cryptography scheme chosen allows for exactly that while it is also possible to use the same scheme for classical encryption problems. Thus, the here proposed quantum onion routing scheme is also more easily compatible with larger systems requiring both classical and quantum cryptographic schemes. 

Incidentally, this also allows us to further develop the intersection of post-quantum cryptography with quantum computing. A secondary goal of this paper is therefore to implement an isogeny-based post-quantum cryptographic scheme within a quantum-computing-based protocol and execution routine.

The paper is structured as follows. In Section~\ref{sec:ClAction} we provide a brief account of the main facts and definitions from the Theory of Complex Multiplication necessary to understand the rational behind our quantum onion routing construction (QOR). Section~\ref{sec:RandomWalks} lifts the problem to the path-finding problem on isogeny graphs and discusses security assumptions. Section~\ref{sec:TorConstruction} connects isogeny graphs to association schemes and their Bose--Mesner algebra, yielding two implementation paths for QOR. The first uses a universal quantum oracle that evaluates the class group action with polynomially many quantum resources. We also outline a second, more intrinsically quantum route based on continuous-time quantum walks (CTQWs)~\cite{Childs,Mull}; this line is developed in a companion work and not pursued further here. Section~\ref{sec:Implementation} presents a fully worked Qiskit example—deliberately small and exponential in resources—whose purpose is to make the mechanics of QOR explicit. Finally, Section~\ref{sec:SecurityFinal} briefly discusses certain security considerations in classical vs.\ quantum settings.

\section{The Ideal Class Group Action}\label{sec:ClAction}
For any given fundamental discriminant $\Delta < 0$, we denote by $K_{\Delta} = \mathbb{Q}(\sqrt{\Delta})$ the corresponding imaginary quadratic number field, and by $Cl(\mathcal{O}_{\Delta})$ the ideal class group of the maximal order $\mathcal{O}_{\Delta}$. We write $h(\D) = |Cl(\mathcal{O}_{\Delta})|$ for the class number. Associated to $Cl(\mathcal{O}_{\Delta})$ is the maximal abelian unramified extension $H_{\Delta}$ of $K_{\Delta}$, known as the Hilbert Class Field of $K_{\Delta}$. 

Fundamental results from Class Field Theory and the Theory of Complex Multiplication yield the following important result: 
\begin{thm}\label{thm:CM} \emph{\cite[\S 10.3, Theorem 5]{Lang}:} For any imaginary quadratic number field $K_{\Delta} = \mathbb{Q}(\sqrt{\Delta})$ with maximal order $\mathcal{O}_{\D}$ and ideal class group $Cl(\mathcal{O}_{\D})$, we let
$\{\mathfrak{a}_{s}\}_{\{1 \leq s \leq h(\D)\}}$ be a set of representatives for the $h(\D)$-many distinct ideal classes of $Cl(\mathcal{O}_{\D})$. Then the numbers $j(\mathfrak{a}_{i})$ are all conjugate over $K_{\D}$ and over $\mathbb{Q}$. For any representative $\mathfrak{a}\in \{\mathfrak{a}_{s}\}_{\{1 \leq s \leq h(\D)\}}$ the following isomorphism holds true
$$Cl(\mathcal{O}_{\D}) \cong Gal(K(j(\mathfrak{a}))/K), \ \text{via the explicit map} \ \
\mathfrak{b} \mapsto \sigma_{\mathfrak{b}}.$$
The irreducible polynomial having the $h(\D)$-many algebraic integers $j(\mathfrak{a}_{s})$ as roots is called the \emph{Hilbert Class Polynomial} associated to the order $\mathcal{O}_{\D}$ and is given by 
\[
H_{\D}(x) = \prod_{1 \leq s \leq h(\D)} \big(x - j(\mathfrak{a}_{s})\big) \in \mathbb{Z}[x].
\]
The action of $Cl(\mathcal{O}_{\D}) \cong Gal(K(j(\mathfrak{a}))/K)$ on the roots of $H_{\D}$ is given by $$\sigma_{\mathfrak{b}}j(\mathfrak{a}) = j(\mathfrak{b}^{-1}\mathfrak{a}),$$ and this action is free and transitive.
\end{thm}

\begin{remark}\label{rmk:ActionCom}
    Let us denote the class group action by `$*$'. As the ideal class group is an abelian group, given Theorem~\ref{thm:CM} above, we easily conclude that 
    $$\mathfrak{ab}*j(\mathfrak{c}) = j((\mathfrak{ab})^{-1}\mathfrak{c}) = j((\mathfrak{ba})^{-1}\mathfrak{c}) = \mathfrak{ba}*j(\mathfrak{c}).$$
\end{remark}
The action of the fractional ideals on the roots of $H_{\D}$ defined in the Theorem~\ref{thm:CM} above, is what is referred to in cryptography as \textit{the class group action}. Together with Deuring's Theorems (e.g., \cite{Deuring} and \cite[Theorem 12 \S13.4 and Theorem 14 \S 13.5]{Lang}), which enable the transition from characteristic zero to characteristic $p$ and vice-versa, are the key ingredients that make isogeny-based cryptography built on the ideal class group action, possible. In the ordinary case in particular, where $H_{\D}$ splits completely in $h(\D)$-many distinct roots over $\mathbb{F}_{q}$, an elliptic curve in characteristic zero with $j$-invariant a root of $H_{\D}$ and an elliptic curve in characteristic $p$ with $j$-invariant a root of $H_{\D}\bmod q$ have isomorphic endomorphism rings. Furthermore, the class group action in characteristic zero respects this isomorphism, and we observe the same class group action on the roots of $H_{\D}$ modulo $q$. 

Even in the supersingular case, if instead of the full ring of endomorphisms, which is non-commutative, we consider the subring of $\F_p$-rational endomorphisms, this subring is again an order in an imaginary quadratic field and the ideal class group action applies as above. The advantage is that now the supersingular elliptic curves $E/\F_p$ have known trace $t=0$ and known number of points $|E(\F_p)| = p+1$. Hence, with the right choice of a prime $p$ we can have $\F_p$-rational isogenies of desired degree $\ell$, as long as $\ell | p+1$. That is why the existing cryptographic schemes based on the ideal class group action, such as CSIDH \cite{CSIDH} and SCURF \cite{CSURF}, employ supersingular curves.

The one-to-one correspondence between the classes $[\mathfrak{a}]$ of the ideal class group $Cl(\O_{\D})$ and the homothety classes of lattices $[\Lambda_{\mathfrak{a}}] \cong [\langle 1 , \tau_{\mathfrak{a}} \rangle]$ with $\O_{\D}$ as their full ring of complex multiplication \cite[Corollary 10.20]{Cox} was employed by the authors of \cite[Theorem 3.2.]{JaoExpander} in order to show that the isogeny graph $G_{\D,S,q}(J_{\D},E_{S})$ is isomorphic to the Cayley graph $\G(Cl(\mathcal{O}_{\D}),S)$, for some set $S$ that does not contain the identity element and is closed under inversion. The set of vertices $J_{\D} = \{j_1 , ... , j_{h(\D)}\} \subseteq \mathbb{F}_{q}$ contains the $h(\D)$-many distinct roots of $H_{\D}$ modulo $q$. Every edge $(j_1,j_2) \in E_{S}$ corresponds to an isogeny between elliptic curves $E_{1}, E_{2}$ with corresponding $j$-invariants $j_1, j_2 \in J_{\D}$. The degree of any such isogeny equals some positive integer $b$ such that $b=\text{Norm}_{K_{\D}/\mathbb{Q}}(\mathfrak{b})$ for some ideal $\mathfrak{b} \in S$. When we include in $S$ all ideals of prime norm less than some fixed bound $M \geq (\log |\D|)^{B}$, for some absolute constant $B > 2$ then, assuming GRH, the isogeny graph $G_{\D,S,q}$ becomes an expander graph \cite[Theorem 3.2.]{JaoExpander}.

\section{Walks on Isogeny Graphs}\label{sec:RandomWalks}
\subsection{The Cayley Graph of the Ideal Class Group}\label{subs:Cay}
As the ideal class group of a number field is a finite abelian group and the isogeny graphs are isomorphic to Cayley graphs of the underlying ideal class group, we focus our attention on Cayley graphs of finite abelian groups. For the basic facts and definitions that we quote in the following, we follow mainly \cite{Big} and \cite{GoRo}.  

Let us denote by $\Gamma$ any finite abelian group written multiplicatively, and with identity element $1_{\Gamma}$. Let $1_\Gamma \notin S \subseteq \Gamma$ be any subset that is closed under taking inverses. Then the \emph{Cayley Graph} $\G(\Gamma,S)$ is the graph with vertex set $V(\G) = \Gamma$ and edge set $$E(\G) = E(\G(\Gamma,S)) = \{gh \ | \ hg^{-1} \in S\}.$$

A \emph{subgraph} $\G'$ of $\G$ is a graph with $V(\G') \subseteq V(\G)$ and $E(\G') \subseteq E(\G)$. In the case where $V(\G') = V(\G)$, $\G'$ is called a \emph{spanning subgraph} of $\G$.

The graph $\G(\Gamma,S)$ is an undirected graph with no loops. It is, as expected, vertex-transitive. Recall from above that $G_{\D,S,q}(J_{\D},E_{\D}) \cong \G(Cl(\mathcal{O}_{\D}),S)$ and the ideal class group action on $J_{\D} \equiv V(\G(Cl(\mathcal{O}_{\D}),S))$ is free and transitive. Vertex-transitive graphs are regular; i.e. each vertex has the same number of edges $k \leq |S|$. This integer $k$ is called the \emph{degree}. When we choose $S$ to be a generating set for $\Gamma$, then the graph is connected.    

The \emph{Adjacency} matrix $A(\G)$ of the connected undirected regular graph $\G(\Gamma,S)$ is the $n\times n$ real symmetric matrix, with entries $a_{i,j} = 1$ if $(v_i,v_j) \in E(\G)$ and $a_{i,j} = 0$ if $(v_i,v_j) \notin E(\G)$. As a real symmetric matrix, $A(\G)$ has real eigenvalues $\lambda_n \leq \cdots \leq \lambda_1 = k$, where as above, $k$ is the degree. Each eigenvalue of $A$ can be given in the form
$$\lambda_{\chi} = \sum_{s \in S} \chi(s), \ \chi \in \hat{\Gamma},$$
where as customary $\hat{\Gamma}$ denotes the character group of the finite abelian group $\Gamma$, which is isomorphic (non-canonically) to $\Gamma$ and hence it is of the same cardinality $n = |\Gamma| = |\hat{\Gamma}|$. Thus, the spectrum of $\G(\Gamma,S)$ consists of character sums ranging over the generating set $S$.

\subsection{Random Walks} As we mentioned in Section~\ref{sec:ClAction}, in the case of isogeny graphs $\G_{\D,S,q}$ and for an appropriate set $S$, the result of \cite{JaoExpander} guarantees that these graphs become expander graphs. In terms of the eigenvalues, what this result yields is a bound on the spectral gap. More specifically, the absolute value of every eigenvalue $\lambda \neq \lambda_1$ should be of the order $|\lambda| = O((\lambda_1\log \lambda_1)^{1/2+1/B})$ for some absolute constant $B > 2$, and $M \geq (\log |\D|)^{B}$ is such that $S = S_{M} = S_{M}^{-1}$ contains the ideals of prime norm up to $M$. With such a choice of $S$, there exists a constant $C>0$ such that any walk on
$\mathcal{G}_{\Delta,S,q}$ of length
\begin{equation}\label{eqn:WalkLength}
  t \;\ge\; C\,\frac{\log |Cl(O_{\D})|}{\log\log q}
\end{equation}
is well mixed. 

Assume for simplicity that $Cl(O_{\D}) \cong \mathbb{Z}/r\mathbb{Z}$ is cyclic of odd prime order $r$. In implementation terms, for a walk of length $t$ on the associated isogeny graph, choose some generator $g \in Cl(O)$ and also choose uniformly at random $k = C\lceil \log r \rceil$ exponents
$c_1,\ldots,c_k \in (\mathbb{Z}/r\mathbb{Z})^{\times}$, sample a random word
$w=(c_{s_1},\ldots,c_{s_m})$ with $s_\ell \in \{1,\ldots,k\}$, and form the
class-group element
\begin{equation}\label{eqn:RandomElement}
  g^{\,e}, \qquad e \equiv \sum_{\ell=1}^{m} c_{s_\ell} \pmod r.
\end{equation}
This random element $g^e$ is then used for the class group action on the given $j$-invariant, say $j_0$, to yield the end-point $g^e * j_0 = j_1$ of the walk.

We rely on the \emph{vectorization} (group-action inversion) problem for the ideal class group action, which is believed to remain hard even for quantum adversaries. Concretely: given a pair of $j$-invariants $(j_0,j_1)$, it is not possible in quantum polynomial time to recover the secret $g^e\in \mathrm{Cl}(\mathcal{O}_\Delta)$ such that $g^e * j_0 = j_1$. Equivalently, one cannot efficiently reconstruct a path (walk) on the underlying isogeny graph $\mathcal{G}_{\Delta,S,q}$ connecting $j_0$ to $j_1$.

In Section~\ref{subs:AS-BM} we describe these isogeny graphs via association schemes and their Bose–Mesner algebra, with emphasis on the cyclic, prime-order case; in this setting the graph is a disjoint union of undirected $r$-cycles. This perspective both places the path-finding problem in an algebraic–spectral context and sets the stage for a parallel line of work in which we exploit this structure to implement our quantum onion routing via CTQWs.

\section{Quantum Onion Routing based on the ideal class group action}\label{sec:TorConstruction}
The main goal of this paper is to construct a quantum onion routing scheme based on the ideal class group action. Without loss of generality, we present our construction in the ordinary case and, as we already discussed in Section~\ref{sec:ClAction}, the same ideas readily extend to the supersingular setting, in the same way as CSIDH and CSURF extended the original CRS in the Diffie–Hellman context.

Given $\D$, $p$ and $a >0$, the ideal class group $Cl(\O_{\D})$ can be computed in quantum-polynomial time~\cite[Theorem 3]{Hall}. We choose $\D$ so as $Cl(\O_{\D})$ has a large-enough, for security reasons, cyclic part. Such discriminants are not hard to find since, according to the Cohen-Lenstra Heuristics \cite{C-L}, the odd part of the class group of imaginary quadratic number fields is cyclic with probability $\sim 97.7575\ldots \%$. 

In order to simplify our scheme and notation, from now on we will assume further that $Cl(\mathcal{O}) \cong \mathbb{Z}/r\mathbb{Z}$, for $r$ a large prime of exponential size, and that $a=1$; i.e. $q=p$. The set 
$$J = \{j_i : 0 \leq i \leq r-1\} \subset F_p$$ contains the $r = |Cl(\mathcal{O})|$-many $j$-invariants. The corresponding isogeny graph is isomorphic to the Cayley graph $\G_r \coloneqq \G(Cl(O_{\D}),S) \cong \G(\mathbb{Z}/r\mathbb{Z},S)$, for a spanning set $S \subset Cl(O_{\D})$; i.e. we include in $S$ all ideals of prime norm less than some fixed bound $M \geq (\log |\D|)^{B}$, for some absolute constant $B > 2$, as we discussed in the end of Section~\ref{sec:ClAction}. 

In this prime-order case, every non-trivial subgraph is a spanning subgraph of $\G_r$, and a random element $g^e$, as in Equation~\ref{eqn:RandomElement}, is a walk on such a graph comprised of $O(\log r)$-many $r$-cycles. As the ideal class group action is commutative \ref{rmk:ActionCom} and injective (see, e.g. \cite[Lemma 5.1]{ChJaSo}), layered encryptions are possible also in the quantum setting, as reversibility is guaranteed.

A naive sketch of the protocol, with only three users present for simplicity, is outlined below in Procedure~\ref{alg:Naive}. These three users will be called Alice (sender), Bob (intermediary), and Carol (receiver).
\begin{figure}[h!]
\captionsetup{name=Procedure}
    \centering
        \caption{High-Level Overview of the First Part of the Protocol}\label{alg:Naive}
        \
    \begin{enumerate}
    \item Fix and make public a $j$-invariant $j_0 \in J_{\D}$.
    \item When Carol is notified by Bob that somebody wants to send her a meassage, she computes her secret element $\mathfrak{c} \in Cl(O_{\D})$ as described in Equation~\ref{eqn:RandomElement} and sends to Bob $$\mathfrak{c}*j_0 = j_C.$$
    \item Bob encrypts with his secret $\mathfrak{b} \in Cl(O_{\D})$ and sends to Alice
    $$\mathfrak{b}*j_C = j_{BC} = j_{CB}.$$
    \item Alice does the same with her secret $\mathfrak{a} \in Cl(O_{\D})$ and sends back to Bob  $$\mathfrak{a}*j_{BC} = j_{ABC}.$$
    \item Bob removes his encryption by applying  
    $$\mathfrak{b}^{-1}*j_{ABC} = j_{AC}.$$
    \item When Carol receives $j_{AC}$ from Bob, she undoes her encryption as well, and obtains Alice's key $j_A$, which she will now use to read Alice's message.
\end{enumerate}
\end{figure}

In terms of implementation, there are two fundamental issues. Firstly, we need to specify how to apply the class group action $*$ and how to deal with the exponentially large isogeny graph. Secondly, we need to specify where and how to hide Alice's message. The subsequent sections are devoted to answering these questions.

\subsection{The Message Encryption}\label{sec:Message} Let $|m\rangle\in\mathcal{M}$ be the quantum message in the message space $\mathcal{M}$ of dimension $2^N$. To each $j$-invariant $j \in J_{\Delta} \subseteq \mathbb{F}_p$, which can be represented by $\log p$-many qubits, we can associate a quantum circuit $C(j)$ on $N$-qubits as follows:    
\begin{equation}\label{Circuit}C(j)=\bigotimes_{k=0}^{N-1}R_X(\vartheta_{k}(j)).\end{equation}
For the angles $\vartheta_{k}(j)$, we take a fixed $N$-dimensional (scrambled) Sobol' sequence $S$, choose Sobol' point $S_{2^{K}+j}$, and $\vartheta_k(j)$ is $2\pi$ times the $k$-th coordinate of the Sobol' point. Since the Sobol' points are more uniform, the potential rotation angles are uniformly distributed in the parameter space. To make this as uniform as possible, we want the band $[2^K,2^K+p]$ to be as close to $[2^K,2^{K'}]$ as possible; i.e., $p\approx 2^K$ gives the band $[2^K,2^{K+1}]$, which is the case for our cryptographic construction since $p$ needs to be of exponential size.

\begin{remark}
    It should be noted that the explicit choice of $R_X$ gates with Sobol' sequence generated angles is not critical for the implementation. Any suitable and efficiently implementable circuit construction from the $j$-invariants would be viable. For example, one could choose $U$-gates instead and a Sobol' sequence of size $3N$. Adding entangling gates may also be considered.
\end{remark}

Going back to step~(4) of the naive Procedure~\ref{alg:Naive} above, Alice could send to Bob 
$$(j_{ABC},\C(j_A)|m\rangle),$$ where $|m\rangle \in \mathcal{M}$ is her secret message. Bob can undo his encryption and send 
$$(j_{AC},\C(j_A)|m\rangle)$$ to Carol, and once Carol obtains $j_A$, then she can read Alice's message by using the inverse circuit, i.e. $$C(j_A)^{-1}C(j_A)|m\rangle.$$ 
We notice however that the hidden message $\C(j_A)|m\rangle$ remains the same throughout the channel. Even though quantum data cannot be cloned nor observed without tampering, one can still make imperfect copies of this message and, in case they obtain a big overlap, then the communications channel will be reveled. To address this issue, we propose the following variant that makes use of an additional Diffie–Hellman setup.

Every user in the network has an established shared secret key with their neighbor (either via quantum or classical computation). The secret shared key is another $j$-invariant derived through Diffie–Hellman using the class group action. In this scenario, the exchange can proceed as shown in Procedure~\ref{alg:Message}:
\begin{figure}[h!]
\captionsetup{name=Procedure}
    \centering 
        \caption{Hiding the Quantum Message}\label{alg:Message}  \   
\begin{enumerate}
    \item Alice sends $(j_{ABC}, C(j_{ab})C(j_A)|m\rangle)$ to Bob.
    \item Bob removes the encryption via the inverse circuit $C(j_{ab})^{-1}$ associated with his secret shared key $j_{ab}$ with Alice, then adds the encryption $C(j_{bc})$ associated with his shared secret key $j_{bc}$ with Carol and sends her the following $(j_{AC}, C(j_{bc})C(j_A)|m\rangle)$.
    \item By removing the encryptions, Carol can read $(j_A,C(j_A)|m\rangle)$ and she can now obtain the message $m$ by applying the inverse circuit $C(j_A)^{-1}$.
\end{enumerate}
\end{figure}

In the next section we will define Association schemes with respect to isogeny graphs - a viewpoint that clarifies and streamlines the implementation of the `quantum' class group action.

\subsection{Cyclic Association Schemes and their Bose-Mesner Algebra}\label{subs:AS-BM} In this section we collect basic facts and definitions for cyclic association schemes and their Bose-Mesner Algebra, noting that many of these facts may still hold for association schemes in general. For more details the interested reader may refer to \cite{Godsil}, or to \cite{Mull} for a more concise account.

For any integer $n \geq 3$ we denote by $\C_{n}$ the cycle-graph of length $n$. It is a Cayley graph over $\mathbb{Z}/n\mathbb{Z}$ with generating set $\{1,n-1\}$ (modulo $n$). We let $d = \lfloor \frac{n}{2} \rfloor$ and for $0 \leq r \leq d$ we consider the following adjacency matrices $A_{r}$, where the indices of the matrices are computed modulo $n$: 
$$(A_{r})_{j,k} = \begin{cases} 1; \ \text{if} \ j-k \in \{r,-r\} \\
0; \ \text{o.w.}
\end{cases}.$$
For any matrix $M$ we denote by $M^{t}$ its transpose. We denote by $A_{0}$ the $n\times n$ identity matrix and by $J = J_{n}$ the all-ones matrix. We notice that $A_{1}$ is the adjacency matrix of $C_{n}$. The set of matrices
$$\A_{n} = \{A_{0},\ldots,A_{d}\}$$ can be proven to form a special type of a \emph{symmetric $d$-class association scheme} (e.g. \cite[pp.19-20]{Mull}), known as the \emph{cyclic association scheme of order $n$}. 

\begin{defn}\label{defn:AS}
    A set of $n \times n$-dimensional \emph{symmetric} $01$-matrices $\A=\{A_{0},\cdots,A_{d}\}$ is called a $d$-class symmetric association scheme if
    \begin{enumerate}
        \item $A_{0} = I_{n}$
        \item $\sum_{j=0}^{d} A_{j} = J_{n}$
        \item $\forall A_{j} \in \A, A_{j}^{t} \in \A$
        \item There are integers $p_{ij}(k)$, known as \emph{intersection parameters}, such that $$A_{i}A_{j} = \sum_{k=0}^{d}p_{ij}(k)A_{k}, \forall \ 0 \leq i, j \leq d.$$ 
        \item For $0 \leq j \leq d$ in particular, we have that $n_{j} \coloneqq p_{j,j}(0)$ is the degree of the regular graph corresponding to the adjacency $A_{j} \in \A$.  
    \end{enumerate}
\end{defn}

{\em Cyclic} symmetric association schemes in particular, enjoy further properties, which we summarize in the following remark.
\begin{remark}
    As a symmetric cyclic association scheme, $\A_{n}$ has the following properties:
    \begin{enumerate}
        \item From Definition~\ref{defn:AS} one can deduce that $A_{i}A_{j} = A_{j}A_{i}$ for all $0 \leq i,j \leq d$.
        \item Cyclic Association Schemes are distance-regular graphs.
        \item We let $\C \equiv \C_{n}$ denote the adjacency matrix of the directed cycle of order $n$. $\C$ is of the form $(\C)_{i,j} = 1$ if $j-k = 1$ and $(\C)_{i,j} = 0$, otherwise. Then every $A_{j} \in \A$ can be written in terms of $\C$ as follows:
        $$A_{j} = \begin{cases} \C^{j} + \C^{-j}; \ \text{if} \ j \notin \{0,n/2\} \\
        \frac{1}{2}(\C^{j}+\C^{-j}); \ \text{if} \ j \in \{0,n/2\}
        \end{cases}.$$
        \item The intersection numbers of $\A_{n}$ are explicitly known:
        $$p_{i,j}(k) = \begin{cases}
            1; \ \text{if} \ i-j \equiv \pm k \bmod n \ \text{or} \ i+j \equiv \pm k \bmod n\\
            0; \ \text{o.w.}
        \end{cases}.$$
    \item For every matrix $M \in \mathbb{C}[\A_{n}]$ there exists unique polynomial $f(x) \in \mathbb{C}[x]$ of degree $deg(f(x)) \leq n-1$ such that $M = f(\C)$.
    \item Let $\omega$ be a fixed primitive $n$th root of unity; for example $\omega = e^{2\pi i/n}$. The explicit eigenvalues of $\A_{n}$ are given as
    $$p_{j}(r) = \omega^{jr} + \omega^{-jr}.$$
    \end{enumerate}
\end{remark}

Given a symmetric association scheme $\A$, Definition~\ref{defn:AS} implies that the algebra $\mathbb{C}[\A]$, known as the \emph{Bose-Mesner} algebra of $\A$, is a commutative algebra of dimension $d+1$. It is a known result that all matrices of $\A_d$ are simultaneously diagonalizable (see, e.g. \cite{Godsil}), hence every $A \in \A_d$ can be written as 
\begin{equation}\label{eqn:decompn} A= \sum_{1\leq s \leq d} \theta_s E_s \ \in \A_d,\end{equation}
where the set 
$\{E_0, E_1, \ldots, E_d\}$ is the set of \emph{spectral idempotents} of the association scheme. In the special case of a cyclic symmetric association scheme, each $E_s$ is is given by
\[
E_s = \begin{cases}
    \frac{1}{n}\sum_{k=0}^{n-1} (\omega^{kr} + \omega^{-kr}) \C^k; \ 1 \leq s \leq \lfloor n/2 \rfloor \\[0.3cm]
    \frac{1}{n} \sum_{k=0}^{n-1} \omega^{kr} \C^{k}; \ s \in \{0, n/2\}.
\end{cases}
\]

An isogeny graph $\mathcal{G}_r=\mathcal{G}_{\Delta,S,q}$ with ideal class group
$\mathrm{Cl}(\mathcal O_\Delta)\cong \mathbb{Z}/r\mathbb{Z}$ (with $r$ an odd prime) can be written
as the edge–disjoint union of $|S|/2$ undirected $r$-cycles. In particular, if $|S|=\Theta(\log r)$, the graph is a union of
$\Theta(\log r)$-many $r$-cycles $A_s \in \A_d$. A walk on $\mathcal{G}_r$ corresponding to the element $g^e$ of \eqref{eqn:RandomElement} is a
sequence of $m$-many steps $c_s$, each one of them carried along one of these $r$-cycles. 

In the next section we analyze these isogeny graphs through the lens of association schemes and their
Bose–Mesner algebra, and we discuss isogeny graphs with respect to a purely quantum form of computing, namely that of continuous-time quantum walks (CTQW), before we continue to the implementation of our QOR via a universal quantum oracle, which we define in Section~\ref{sec:UOracle}. 

\subsection{Isogeny Graphs and Unitaries}\label{sec:GrUni}
Let us fix a fundamental discriminant $\D$ and a good prime $p$ that splits the Hilbert Class polynomial completely into $h(\D)$-many distinct roots over $\mathbb{F}_{q}$, $q = p^a$ for some $a \in \mathbb{N}$. Any subgraph of $\G \equiv \G(Cl(\O_{\D}),S)$ has an associated adjacency matrix $A$, which forms the Hamiltonian of the quantum system determined by the unitary operator 
$$U_{A}(t) = e^{-itA},$$ also known as the time evolution operator. The time evolution is described by 
$$|\psi(t)\rangle = e^{-itA}|\psi(0)\rangle,$$ 
where we follow the physicists’ notation and write column vectors $\psi \in \mathbb{C}^{n\times1}$ as kets $|\psi\rangle$. 

The unitary $U_A(t)$ corresponds to a walk on the graph with adjacency $A$, for time $t$. In isogeny-based schemes this would correspond to the secret walk of a user on the isogeny graph with adjacency $A$. As the underlying ideal class group action is commutative~\ref{rmk:ActionCom}, the walks are also expected to commute and hence the unitaries are expected to commute. This claim can also be verified rigorously, since one can show that the unitaries lie in the commutative Bose-Mesner Algebra. 
\begin{lemma}
    Let $\A_d$ denote the symmetric association scheme corresponding to the Cayley graph of a finite cyclic group isomorphic to $\mathbb{Z}/n\mathbb{Z}$, and for any $A \in \A_d$ consider the corresponding unitary $U_A(t) = e^{-itA}.$ Given any $A, B \in \A_d$, we have that the corresponding unitaries commute; i.e.
    $$U_A(t_A)U_B(t_B) = U_B(t_B)U_A(t_A).$$
\end{lemma}
\begin{proof}
  For any $A \in \A_d$ it follows easily via Equation~\ref{eqn:decompn} (see, e.g. \cite[Lemma 3.2.4]{Mull}), that the corresponding unitary belongs to the Bose-Mesner algebra
\begin{equation}\label{eqn:UnitaryBM} 
U_A(t) = \sum_{1 \leq s \leq d} e^{i t \theta_{A,s}} E_s \ \in \mathbb{C}[\A_d].
\end{equation}  
Since the algebra is commutative, the claim follows.
\end{proof}

One can now express the walk on the isogeny graph induced by the action of the element $g^e, \ e \equiv \sum_{\ell=1}^{m} c_{s_\ell} \pmod r$ of Equation~\ref{eqn:RandomElement}, in terms of unitaries, as follows:
\begin{equation}\label{eqn:AdjacencyElement}
     g^e * j_0 \equiv \prod_{\ell = 1}^{m}U_{c_{s_{\ell}}}(t_{s_{\ell}}) |j_0\rangle.
\end{equation}

Further study of the isogeny walk via a CTQW is being carried out in a parallel project. In the next section we assume the existence of a universal oracle for the class group action and describe in detail how the QOR protocol runs under this assumption, leaving the discussion for the quantum oracle in the last subsection, Section~\ref{sec:UOracle}.

\subsection{Isogeny graphs in superposition}\label{sec:superpn}
As above, we assume for simplicity that $Cl(\mathcal{O}) \cong \mathbb{Z}/r\mathbb{Z}$, for $r$ a large prime of exponential size. The set
$$J = \{j_i : 0 \leq i \leq r-1\} \subseteq \mathbb{F}_p$$ contains the $r$-many $j$-invariants associated with $Cl(\mathcal{O})$, and $j_{0} \in J$ is a fixed public $j$-invariant. The corresponding cyclic association scheme $\mathcal{A}_{d}$, is of order $d = \frac{r-1}{2}$, also of exponential size.
\\
\begin{center}
    \textbf{Procedure 3.} Isogeny Graphs in Superposition and the QOR
\end{center}
\vspace{0.25cm}
When Bob notifies Carol that somebody wants to send her a message,
\begin{enumerate} 
\item Carol chooses and keeps secret a random element $\mathfrak{c} = g_C^{e_C} \in \mathrm{Cl}(\mathcal{O})$, constructed as discussed above. The associated isogeny graph corresponding to the action of $\mathfrak{c}$ on the $j$-invariants of the set $J$ is represented by an adjacency matrix, which we denote by $C$, and $C \in \A_d$. 
    \item Carol furthermore chooses two secret values $\omega,\Omega\in\mathbb{Z}/r\mathbb{Z}$ satisfying $0\le\omega<r$ and $0\le\Omega<r-\omega$. Here we assume both $\omega$ and $\Omega$ to be of polynomial size, although this restriction can be lifted relative to a ``global mapper oracle'', cf. Section~\ref{sec:UOracle}. Carol prepares a uniform superposition over $i \in \mathbb{Z}/\Omega\mathbb{Z}$
    \[
    \frac{1}{\sqrt{\Omega}} \sum_{i=0}^{\Omega-1} |i\rangle.
    \]
    \item She then loads $\ket{j_0}$ into another register and applies her unitary operator $U_C$, which we can interpret as a quantum oracle $f_C$,
    that computes in quantum parallelism the class group action, $*$, and gives the resulting $j$-invariant
     \[
    |i\rangle\otimes|j_0\rangle \mapsto |i\rangle\otimes|f_C(j_0)\rangle = |i\rangle\otimes|\mathfrak{c}^1 * j_0\rangle = |i\rangle\otimes|j_1\rangle.
    \]
    Applying the oracle $\omega$ times, yields the resulting state  
    \[
    \frac{1}{\sqrt{\Omega}} \sum_{i=0}^{\Omega-1} |i\rangle\otimes|j_\omega\rangle.
    \]
    \item Finally, Carol applies a ``mapper oracle'' as shown in Figure~\ref{fig:mapper}. This mapper oracle\footnote{The requirement for $\Omega$ to be of polynomial size implies that this mapper oracle is constructable in polynomial time from her shift oracle $U_C$. As $U_C$ needs to be compiled into a valid gate sequence, each gate in the $U_C$-sequence can be replaced with their controlled version to obtain a gate sequence for the controlled-$U_C$-gate.} applies her oracle $i$ many times if the control register is in the state $\ket i$. The resulting state is 
    \[
    \frac{1}{\sqrt{\Omega}} \sum_{i=0}^{\Omega-1} |i\rangle\otimes|j_{i+\omega}\rangle.
    \]
    
    If possible in practical implementations, Carol should retain the index-register (containing the $\ket i$ factor) at all times. However, the communication scheme remains secure if both the index- and $j$-register qubits are sent during the protocol. For simplicity, we will assume that Carol retains the index-register.
\item Carol sends the quantum state $$|\psi_0\rangle = \frac{1}{\sqrt{\Omega}}\sum_{i=0}^{\Omega-1}|j_{i+\omega}\rangle,$$ which corresponds to a part of the isogeny-cycle in superposition produced by the element $\mathfrak{c} \in Cl(O_{\D})$ acting on the set $J$, to Bob.

\item In the same way, Bob chooses $\mathfrak{b} = g_B^{e_B} \in Cl(O_{\D})$, which corresponds to his adjacency $B \in \A_d$. His associated quantum oracle $f_B$ computes the class group action, i.e.
$$f_B(j) = \mathfrak{b}*j,$$ and Bob computes and sends to Alice the following state
\[
|\psi_1\rangle = \frac{1}{\sqrt{\Omega}} \sum_{i=0}^{\Omega-1} |f_B(j_{i+\omega})\rangle.
\]
\item Similarly, Alice chooses her secret $\mathfrak{a} = g_A^{e_A} \in Cl(O_{\D})$, corresponding to some adjacency matrix $A \in \A_d$, and with her quantum oracle $f_A$ computes her class group action. She also attaches the message $|m\rangle$ which is encrypted with the circuit $C(j_A)$, $j_A = \mathfrak{a} * j_0$, (preventing Bob from reading the message) and the transport encryption layer $C(j_{ab})$ (preventing tracking of the message). Thus, she sends the following state
\[
|\psi_2\rangle = \frac{1}{\sqrt{\Omega}} \sum_{i=0}^{\Omega-1} |f_A(f_B(j_{i+\omega}))\rangle \otimes C(j_{ab})C(j_A)|m\rangle
\]
 back to Bob.
\item  Bob undoes his oracle $f_B$ on the $j$-invariants register, modifies the message as discussed in Section~\ref{sec:Message}, and obtains
\[
|\psi_3\rangle = \frac{1}{\sqrt{\Omega}} \sum_{i=0}^{\Omega-1} |f_A(j_{i+\omega})\rangle \otimes C(j_{bc})C(j_A)|m\rangle
\]
which he sends on to Carol.

\item  Carol  removes the transport encryption layer $C(j_{bc})$ from the message and splits the message register off. Combining the $j$-invariants register with the index-register again, she obtains two separate quantum states; the still encrypted message $C(j_A)|m\rangle$ and the key information state
    \[
    |\psi_4\rangle=\frac{1}{\sqrt{\Omega}} \sum_{i=0}^{\Omega-1} |i\rangle\otimes|f_A(j_{i+\omega})\rangle.
    \]
\item Carol now measures the key information state $|\psi_4\rangle$ and obtains a random state $|i\rangle\otimes|f_A(j_{i+\omega})\rangle$. From this, she can compute $j_A=\mathfrak{c}^{-i-\omega}f_A(j_{i+\omega})$ either classically or using her oracles (as we will in Section~\ref{sec:Implementation}).\footnote{With Carol's mapper oracle and repeated application of her shift oracle $U_C$, she can also compute $\frac{1}{\sqrt{\Omega}} \sum_{i=0}^{\Omega-1} |i\rangle\otimes|\mathfrak{c}^{-i-\omega}f_A(j_{i+\omega})\rangle=\sum_{i=0}^{\Omega-1} |i\rangle\otimes|j_A\rangle$ without measurement and then simply measure the $j$-register.}
\item Finally, Carol can now apply $C(j_A)^{-1}$ to the message register and retrieve Alice's secret message $|m\rangle$. 
\end{enumerate}
    \begin{figure}[h!]
        \centering
        \begin{align*}
            \Qcircuit @C=1em @R=.7em {
              \lstick{\text{index reg }q_0} & \ctrl{3} & \qw & \qw & \qw & \qw & \qw & \qw & \qw \\
              \lstick{\text{index reg }q_1} & \qw & \ctrl{2} & \ctrl{2} & \qw & \qw & \qw & \qw & \qw \\
              \lstick{\text{index reg }q_2} & \qw & \qw & \qw & \ctrl{1} & \ctrl{1} & \ctrl{1} & \ctrl{1} & \qw \\
              \lstick{j\text{-reg}} & \gate{U_C} & \gate{U_C} & \gate{U_C} & \gate{U_C} & \gate{U_C} & \gate{U_C} & \gate{U_C} & \qw \\
            }
        \end{align*}
        \caption{Example Carol's of the ``mapper oracle'' for a three-qubit index register. In general, this simple mapper oracle implementation requires $2^{\lceil\log_2\Omega\rceil}-1$ controlled shift oracles~$U_C$.}
        \label{fig:mapper}
    \end{figure}

\begin{remark}
    It should be noted that the message encryption via $C(j_A)\ket m$ could also be replaced with a quantum public key encryption system. In that case, the procedure required to communicate $j_A$ from Alice to Carol is not necessary as Carol retains her private key and the cryptosystem is asymmetric, i.e., cannot be broken based on the public knowledge of the public key. However, that also requires combining two separate encryption systems within the quantum onion routing communications protocol, which may or may not be wise. There are multiple dimensions to consider here. On the one hand, having two different underlying hard mathematical problems means that if one is broken, the message may still be safe. On the other, having compatible schemes is desirable in terms of design principles, as is also evident in our construction which combines circuits of $j$-invariants resulting from both the Diffie-Hellmann as well as the layered encryptions. At the same time however, using the same underlying problem for both schemes can compromise the quantum onion routing protocol because breaking that underlying problem reveals everything, while using different underlying problems means that if the $C(j_A)$ layer is broken, then any intermediary (Bob) can read the message, while breaking the Diffie-Hellman layer allows for tracking of the communication; thus, doubling the failure points.
\end{remark}

\subsection{A Universal Quantum Oracle}\label{sec:UOracle}
The ideal class group $Cl(O_{\D})$ of an imaginary quadratic number field is isomorphic to the so-called \textit{form class group} of primitive positive definite binary quadratic forms of discriminant $\D$ (see, e.g. \cite[Theorem 5.30]{Cox}). The explicit isomorphism allows for a straightforward transition between the elements of the two groups. As a result, one can work in the form class group instead, which yields a very efficient arithmetic in the group, when considered as composition of forms. An algorithm of Sch\"onhage \cite{Sch} gives the time for reduction as well as composition of forms to be \begin{equation}\label{eqn:time1}O(m(|\D|) \log |\D|) \sim O(\mathrm{poly} \log |\D|),\end{equation}
where $m(|\D|) \leq \log(|\D|)^2$ denotes the time for multiplication of two integers of size $\sim |\D|$. 

As above, we assume that the class group $G$ is cyclic of prime order $r$. Given any $g \in G$ and $1\leq s \leq r$, we can perform exponentiation $g^s$ by repeated squaring (see the Circuit in Figure~\ref{fig:RSq}), and together with equation~\ref{eqn:time1} above, yield the total time \begin{equation} \label{eqn:time2} O(\log s \ \mathrm{poly} \log |D|) \sim O(\mathrm{poly} \log r).\end{equation}
     
          \begin{figure}[h!]
        \centering
      \begin{align*}
        \Qcircuit @C=1em @R=.7em {
          \lstick{\text{index reg }q_0} & \ctrl{7} & \qw & \qw & \qw & \qw & \qw & \qw & \qw \\
          \lstick{\text{index reg }q_1} & \qw & \ctrl{6} & \qw & \qw & \qw & \qw & \qw & \qw \\
          \lstick{\text{index reg }q_2} & \qw & \qw & \ctrl{5} & \qw & \qw & \qw & \qw & \qw \\
          \lstick{\text{index reg }q_3} & \qw & \qw & \qw & \ctrl{4} & \qw  & \qw & \qw & \qw \\
          \lstick{\text{index reg }q_4} & \qw & \qw & \qw & \qw & \ctrl{3} & \qw & \qw & \qw  \\
          \lstick{\text{index reg }q_5} & \qw & \qw & \qw & \qw & \qw & \ctrl{2} & \qw & \qw \\
          \lstick{\text{index reg }q_6} & \qw & \qw & \qw & \qw & \qw & \qw & \ctrl{1} & \qw \\
          \lstick{j\text{-reg}} & \gate{U_g} & \gate{U_{g^2}} & \gate{U_{g^4}} & \gate{U_{g^8}} & \gate{U_{g^{16}}} & \gate{U_{g^{32}}} & \gate{U_{g^{64}}} & \qw \\
        }
      \end{align*}
     \caption{Circuit for repeated squaring for the universal quantum oracle. This circuit applies $g^s$ to the state in the $j$-register depending on the binary decomposition of $s$ stored in the index register.}
      \label{fig:RSq}
    \end{figure}
    
Denote by $U(g,j)$ the global oracle $$U: Cl(O_{\D}) \times J_\D \rightarrow J_\D$$
$$(g,j) \mapsto g*j,$$
which is injective in each argument given the properties of the class group action, but of course it is not jointly injective.

Each user computes their random element, which is of the form $g^e$ as in Equation~\ref{eqn:RandomElement}, and applies this global oracle accordingly. As the exponent $e$ is of size $O(\log r)$, we need $O(\log r)$-many controlled $U(g^{2^k},-)$ gates, and the total time for this global oracle, given Equations~\ref{eqn:time1}~\&~\ref{eqn:time2}, remains at $O(\mathrm{poly} \log r)$. Like everywhere else in the paper, we do not analyze the runtime of the class-group action per se (or isogeny computation), as implementations are an active research area with steadily improving results.

Now, given such a global oracle, the restrictions on $\omega$ and $\Omega$ to be of polynomial size can be lifted, and arbitrary sub-intervals of the entire cycle corresponding to the chosen element $g^e$ can put in superposition, including the entire cycle itself. With the entire cycle in superposition, the protocol can be formed as shown in Procedure~\ref{alg:superposition}.
\addtocounter{figure}{-1}
\begin{figure}[h!]
\captionsetup{name=Procedure}
    \centering 
        \caption{The Entire Isogeny Cycle in Superposition}\label{alg:superposition}  \   
\begin{enumerate}
    \item Carol prepares the state 
    $$|\psi_0 \rangle = \frac{1}{\sqrt{r}}\sum_{s=0}^{r-1} |s\rangle \otimes |j_0 \rangle.$$
    \item She applies her universal mapper oracle as shown in Figure~\ref{fig:RSq} with $g=\mathfrak{c}$ 
    $$|\psi_1\rangle = \frac{1}{\sqrt{r}}\sum_{s=0}^{r-1} |s\rangle \otimes |\mathfrak{c}^s*j_0 \rangle = \sum_s |s\rangle \otimes |j_s \rangle$$
    and sends out only the second register ($j$-register), retaining the first (index register)
    $$|\psi_2\rangle = \frac{1}{\sqrt{r}}\sum_{s=0}^{r-1} |j_s \rangle.$$
    \item Bob applies the universal oracle $U(\mathfrak{b}, -)$ and sends the following state to Alice
   $$|\psi_3\rangle =  \frac{1}{\sqrt{r}}\sum_{s=0}^{r-1} |\mathfrak{b}*j_s \rangle.$$
   \item Alice applies the universal oracle $U(\mathfrak{a}, - )$ and sends the following state back to Bob
      $$|\psi_4\rangle =  \frac{1}{\sqrt{r}}\sum_{s=0}^{r-1} |\mathfrak{ab}*j_s \rangle.$$
      \item Bob applies the inverse oracle $U(\mathfrak{b}^{-1},-)$ and sends to Carol the following state
    $$|\psi_5\rangle = \frac{1}{\sqrt{r}}\sum_{s=0}^{r-1}|\mathfrak{a}*j_s\rangle.$$
    \item Carol recombines the received $j$-register with the retained index register and applies her inverse universal mapper oracle and obtains
        $$|\psi_{\text{final}}\rangle = \frac{1}{\sqrt{r}}\sum_{s=0}^{r-1} |s\rangle\otimes |\mathfrak{a}*j_0\rangle.$$
        Any measurement will reveal Alice's key $j_A$.
\end{enumerate}
\end{figure}

\begin{remark}
    It should be noted that the oracle in Figure~\ref{fig:RSq} is used in two different ways. For any given $g$ and secret exponent $e$, any actor can load $e$ into the index register which turns the oracle into the shift oracle. For example, if Carol chooses $e_C$, then her shift oracle should act as $\ket j\mapsto\ket{\mathfrak{c}*j}=\ket{g^{e_C}*j}$ which executes with $e_C$ loaded into the index register in Figure~\ref{fig:RSq}. All actors need this application of Figure~\ref{fig:RSq}. However, Carol also needs the mapper oracle to implement the powers of $\mathfrak{c}$ which can be represented as in Figure~\ref{fig:RSq} but with $g=\mathfrak{c}$ and with $i$ loaded into the index register. This executes the mapper operation $\frac{1}{\sqrt{\Omega}}\sum_{i=0}^{\Omega-1}\ket{i}\otimes\ket{j_\omega}\mapsto\frac{1}{\sqrt{\Omega}}\sum_{i=0}^{\Omega-1}\ket{i}\otimes\ket{j_{i+\omega}}$. Carol can implement this mapper oracle either by directly implementing the circuit of Figure~\ref{fig:RSq} with $g=\mathfrak{c}$ and loading $i$ into the index register, or by nesting Figure~\ref{fig:RSq} using one register to hold $i$, one register to hold $e_C$, and constructing the controlled $U_{g^{2^{k}e_C}}$ gates instead of the $U_{g^{2^k}}$ gates in Figure~\ref{fig:RSq}. Instead of nesting, the latter can also be achieved using a quantum arithmetic circuit computing the products $2^ke_C$ and its output register is used as index register for Figure~\ref{fig:RSq}. 
\end{remark}

\section{An Example Implementation with $5$ actors}\label{sec:Implementation}
In this section we present a fully worked, minimal Qiskit example with five actors A, B, C, D, and~E. Here, A wants to send a message to E with B, C, and D acting as intermediaries similar to the classical TOR protocol. For simplicity, we will only cover the key information part of the protocol as the message part is equivalent to classical Diffie-Hellman exchanges and applying simple circuitry for scrambling. We also remove the classical communication overheads of setting up communications channels between actors.

For simplicity, we will implement the class group action oracles by classically pre-computing the corresponding cycles and implementing the corresponding shift operators. The direct implementation of these oracle gates is beyond the scope of this implementation as we want to focus on the communication part of the protocol.

For this example, we choose the ideal class group of the imaginary quadratic number field $\mathbb{Q}(\sqrt{-167})$, which is of order $11$. The corresponding Hilbert Class Polynomial has very large coefficients but it splits completely into eleven linear factors over the finite field $\mathbb{F}_{311}$. The set of $j$-invariants containing its roots is 
$$J = \{307, 248, 236, 223, 213, 209, 193, 182, 116, 12, 1\} \subset \mathbb{F}_{311}.$$ 

This example has five, $\frac{11-1}{2}$, non-trivial undirected cycles, as shown in Figure~\ref{fig:five}, and we attribute a different one as the secret choice made by each actor. The publicly known $j_0$ is chosen as $1$. These computations were performed in PARI/GP \cite{PARI} and SageMath \cite{SageMath}.
\begin{figure}[htbp]
  \centering

  \begin{minipage}{0.9\textwidth} 
    \centering
    \begin{subfigure}{0.44\textwidth}
      \centering
      \includegraphics[width=\linewidth,height=0.22\textheight,keepaspectratio]{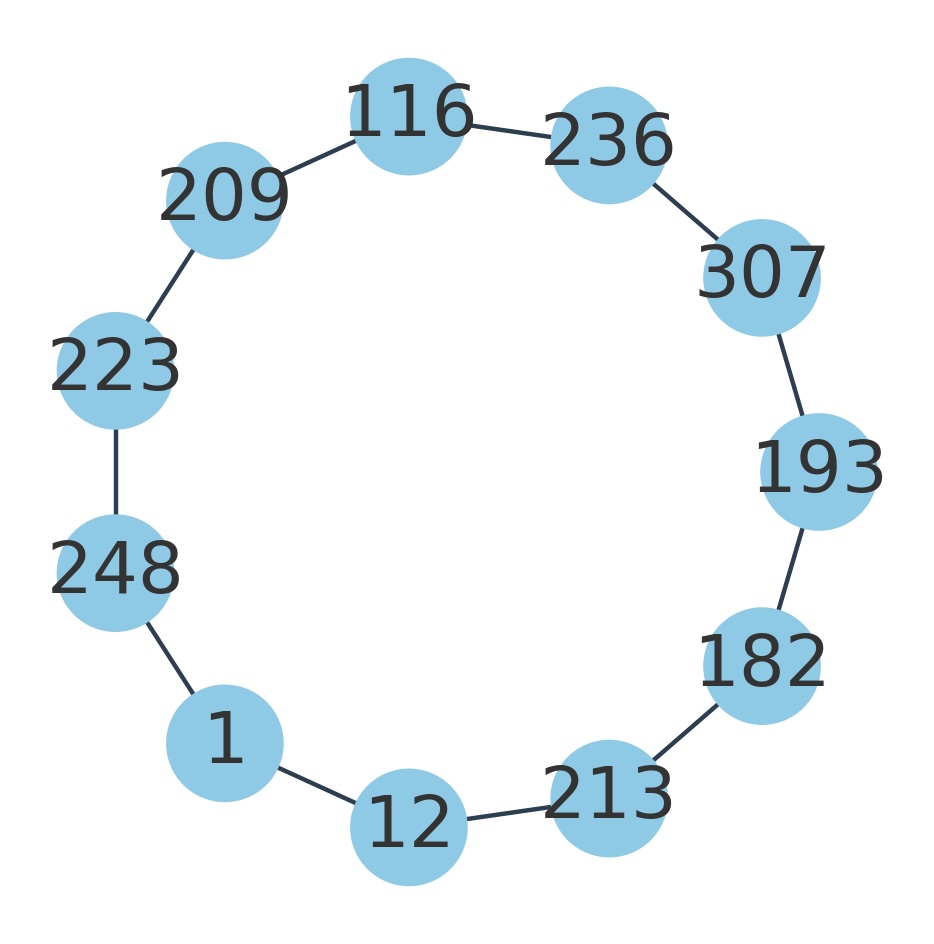}
    \end{subfigure}\hfill
    \begin{subfigure}{0.44\textwidth}
      \centering
      \includegraphics[width=\linewidth,height=0.22\textheight,keepaspectratio]{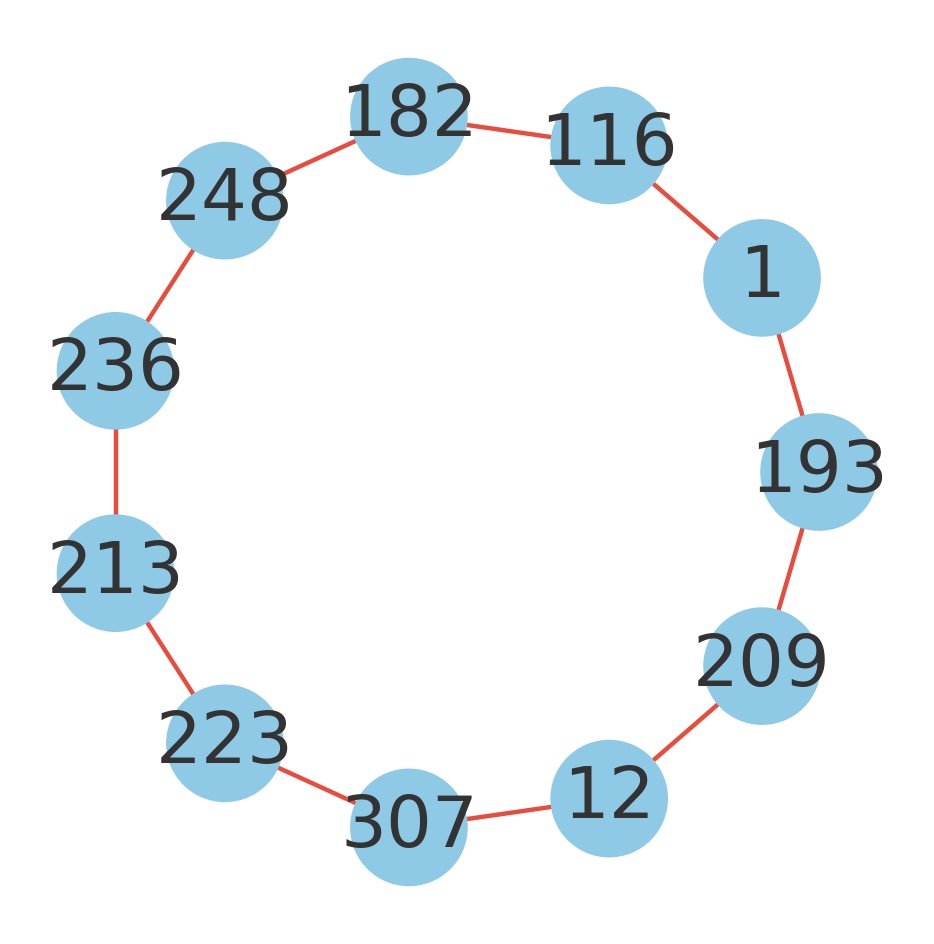}
    \end{subfigure}
  \end{minipage}

  \vspace{0.6em} 

  \begin{minipage}{0.9\textwidth}
    \centering
    \begin{subfigure}{0.29\textwidth}
      \centering
      \includegraphics[width=\linewidth,height=0.22\textheight,keepaspectratio]{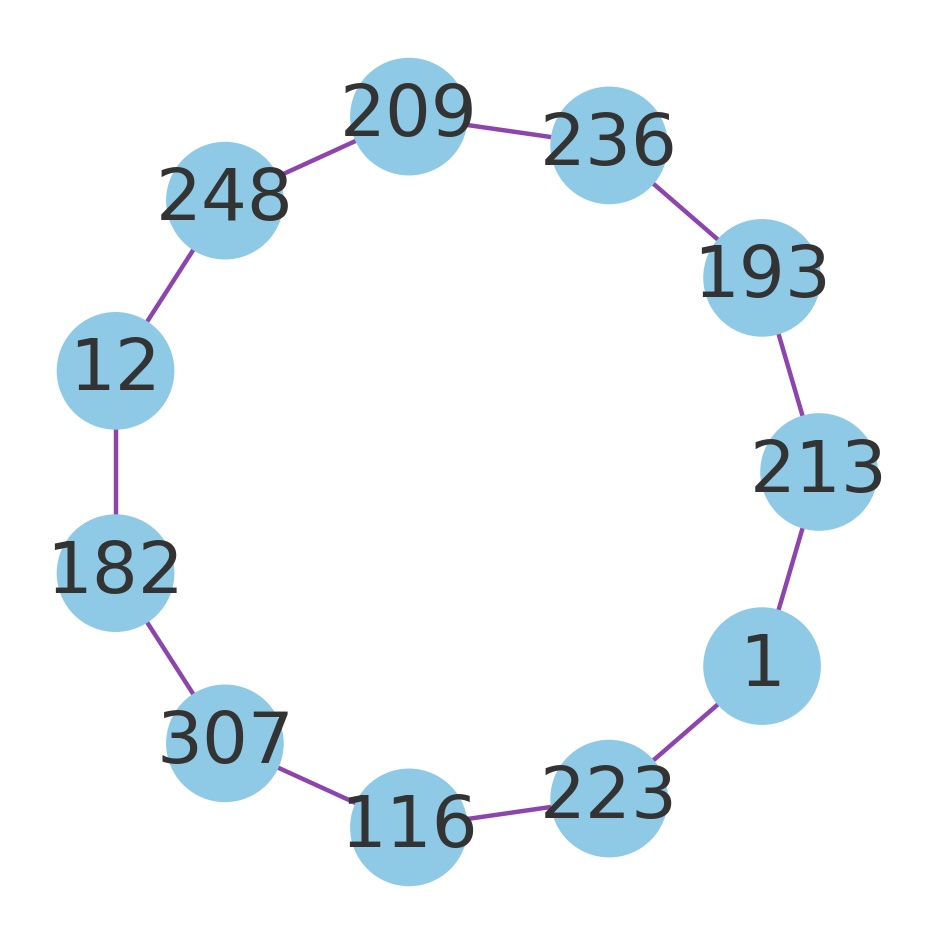}
    \end{subfigure}\hfill
    \begin{subfigure}{0.29\textwidth}
      \centering
      \includegraphics[width=\linewidth,height=0.22\textheight,keepaspectratio]{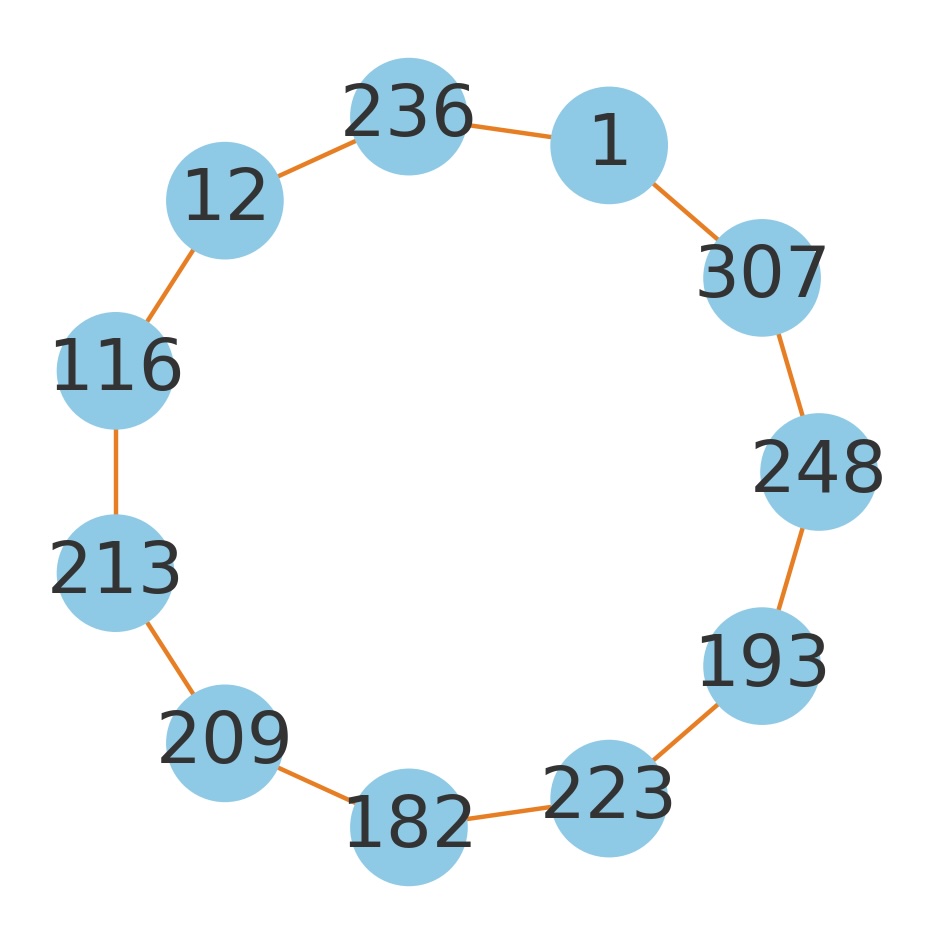}
    \end{subfigure}\hfill
    \begin{subfigure}{0.29\textwidth}
      \centering
      \includegraphics[width=\linewidth,height=0.22\textheight,keepaspectratio]{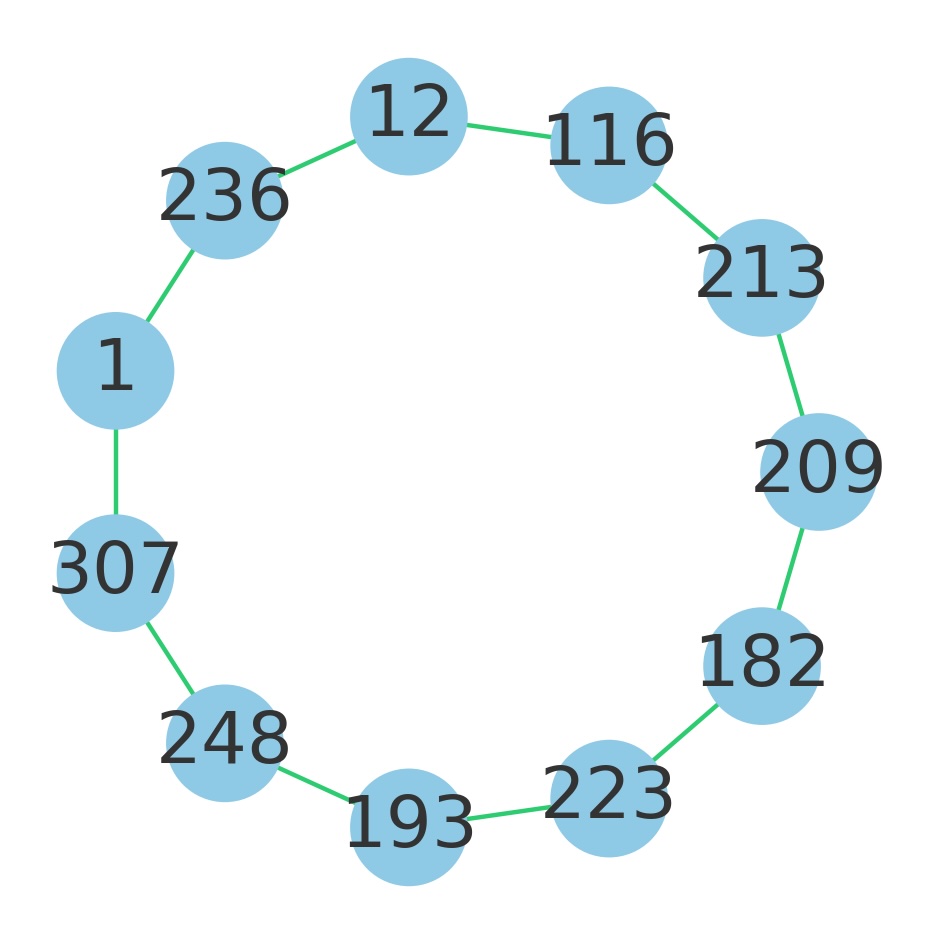}
    \end{subfigure}
  \end{minipage}

  \caption{The five isogeny cycles corresponding to elements of norms $3, 7, 11, 19$ and $97$, in the class group $Cl(O_{-167})$ of order $11$. Each one connects the $j$-invariants in a different way.}
  \label{fig:five}
\end{figure}

For this implementation, we will use Qiskit and load the following libraries.
\begin{minted}[mathescape,
  linenos,
  numbersep=5pt,
  gobble=2,
  frame=lines,
  framesep=2mm,
  breaklines]{python}
    import numpy as np
    import copy
    
    from qiskit import QuantumCircuit, QuantumRegister, ClassicalRegister, transpile
    from qiskit.circuit.library import PhaseGate, UnitaryGate
    from qiskit_aer import AerSimulator
    from qiskit.visualization import plot_histogram
    
    from <private_library> import basic_comparator
\end{minted}
The final library loading \texttt{basic\_comparator} is not publicly available at the time of writing this paper. It supplies the comparator $i<\Omega$ which is required for the oracle in an optional Grover step to filter for only the indices $i<\Omega$. If $\Omega$ is chosen as a power of $2$, this \texttt{basic\_comparator} is not required. Its general implementation has been used previously, cf. e.g. in~\cite{PhysRevX.8.041015,stackexchange.comparator}. 

We may then define the $5$ cycles and give one to each actor. All cycles are represented to start with the public $j$-invariant $j_0=1$. In an actual implementation using oracles implementing the class group action, explicit knowledge of these cycles is not required.
\begin{minted}[mathescape,
  linenos,
  numbersep=5pt,
  gobble=2,
  frame=lines,
  framesep=2mm,
  breaklines,
  firstnumber=10]{python}
    aj = [1,213,193,236,209,248,12,182,307,116,223]
    bj = [1,12,213,182,193,307,236,116,209,223,248]
    cj = [1,116,182,248,236,213,223,307,12,209,193]
    dj = [1,236,12,116,213,209,182,223,193,248,307]
    ej = [1,209,307,213,248,116,193,12,223,236,182]
\end{minted}
This means that the secret key that A needs to communicate to E is $j_A=\mathfrak{a}*j_0$ or \texttt{aj[1]=213}.

We choose $\omega=2$ and $\Omega=5$ which implies that we require $Q=\lceil\log_2\Omega\rceil=3$ many qubits for the index register. Similarly, since the field is $F_{311}$, we require $N=\lceil\log_2311\rceil=9$ many qubits to represent the $j$-invariants. Using the cycles defined above, we can define the oracles that shift one element to the right and act as the identity on non-$j$-invariants for all actors and their inverses (the left shifts) for all but A (A's inverse is not required).
\begin{minted}[mathescape,
  linenos,
  numbersep=5pt,
  gobble=2,
  frame=lines,
  framesep=2mm,
  breaklines,
  firstnumber=15]{python}
    # defining constants and oracle matrices
    omega = 2 # arbitrarily chosen
    Omega = 5 # < number of j-invariants minus omega, chosen arbitrarily
    N = int(np.ceil(np.log2(max(ej))))
    Q = int(np.ceil(np.log2(Omega)))

    A = np.zeros((2**N,2**N))
    B = np.zeros((2**N,2**N))
    C = np.zeros((2**N,2**N))
    D = np.zeros((2**N,2**N))
    E = np.zeros((2**N,2**N))
    
    for j in range(2**N):
        if j in ej:
            A[aj[(aj.index(j)+1)%len(ej)],j] = 1
            B[bj[(bj.index(j)+1)%len(ej)],j] = 1
            C[cj[(cj.index(j)+1)%len(ej)],j] = 1
            D[dj[(dj.index(j)+1)%len(ej)],j] = 1
            E[ej[(ej.index(j)+1)%len(ej)],j] = 1
        else:
            A[j,j] = 1
            B[j,j] = 1
            C[j,j] = 1
            D[j,j] = 1
            E[j,j] = 1

    # define qiskit gates from matrices
    UA = UnitaryGate(A)
    UB = UnitaryGate(B)
    UBinv = UnitaryGate(B.transpose())
    UC = UnitaryGate(C)
    UCinv = UnitaryGate(C.transpose())
    UD = UnitaryGate(D)
    UDinv = UnitaryGate(D.transpose())
    UE = UnitaryGate(E)
    UEinv = UnitaryGate(E.transpose())
\end{minted}
We also need to implement the mapper oracle for E. However, the relatively large value of~$N$ and the fact that we do not have ``real'' circuitry for the oracles \texttt{UE} and \texttt{UEinv}, render the construction of the mapper -- as described in Figures~\ref{fig:mapper} or~\ref{fig:RSq} -- computationally heavy. Instead, we will implement a mapper $\ket i\mapsto \ket{j_i}$ according to E's cycle. Thus, rather than loading $\ket{j_0}$ into the $j$-register and entangling it with the index register via the mapper, we will create entanglement by producing the state $\ket i\otimes \ket i$ and then using the ``fake'' mapper implementation to obtain $\ket i\otimes \ket{j_i}$. This ``fake'' mapper necessarily has super-polynomial complexity as otherwise it could be used to break the cryptosystem.
\begin{minted}[mathescape,
  linenos,
  numbersep=5pt,
  gobble=2,
  frame=lines,
  framesep=2mm,
  breaklines,
  firstnumber=51]{python}
    # mapper implementation not efficient but chosen because turning UE and UEinv into controlled gates is very resource intensive
    ejextended = copy.copy(ej)
    for j in range(2**N):
        if j not in ej:
            ejextended.append(j)
    
    mapper = np.zeros((2**N,2**N))
    for j in range(len(ejextended)):
        mapper[ejextended[j],j] = 1
    
    UM = UnitaryGate(mapper)
    UMinv = UnitaryGate(mapper.transpose())
\end{minted}
We can now set up the quantum circuit in Qiskit.
\begin{minted}[mathescape,
  linenos,
  numbersep=5pt,
  gobble=2,
  frame=lines,
  framesep=2mm,
  breaklines,
  firstnumber=63]{python}
    i_register = QuantumRegister(Q, "i")     # index register
    anc_register = QuantumRegister(1, "anc") # ancilla register for Grover (to make sure indices smaller than Omega)
    j_register = QuantumRegister(N, "j")     # j-invariant register
    
    # corresponding classical registers
    cl_i_register = ClassicalRegister(len(i_register), "cl_i")
    cl_anc_register = ClassicalRegister(len(anc_register), "cl_anc")
    cl_j_register = ClassicalRegister(len(j_register), "cl_j")
    
    # initialize quantum circuit
    qc = QuantumCircuit(
        i_register, anc_register, j_register, 
        cl_i_register,cl_anc_register, cl_j_register
    )
    
    # Step 1: Apply Hadamard to all qubits
    qc.h(i_register)
    qc.barrier()
\end{minted}
All of this is executed by E upon being told that they are to receive a message via the quantum onion routing protocol. It should be noted that the ancilla register is only required for values of $\Omega$ that are not powers of $2$.

At this point, the current quantum state (ignoring the ancilla qubit) is $\frac{1}{\sqrt{8}}\sum_{i=0}^{7}|i\rangle\otimes|0\rangle$. In other words, we have an equal superposition of all computational basis states in the index register, and the $|0\rangle$ state in the $j$-invariants register. We now add the optional Grover step to filter for only the states $i<\Omega=5$.
\begin{minted}[mathescape,
  linenos,
  numbersep=5pt,
  gobble=2,
  frame=lines,
  framesep=2mm,
  breaklines,
  firstnumber=81]{python}
    # Optional gover circuit if Omega is not a power of 2 
    if Omega<2**Q:
        # Filter out numbers larger than Omega-1. (Marking)
        qc.append(basic_comparator(Q, Omega), list(i_register) + list(anc_register))
        # We clearly now have over 50% solutions, due to the way we have encoded the bits.
        
        # So we can use Grover's exact rotation to remove any non-solutions.
        # We know the angle will be:
        theta = np.arccos(1 - (2**Q / (2*Omega)))
        qc.append(PhaseGate(theta), [anc_register[0]])
        
        qc.append(basic_comparator(Q, Omega, reverse=True), list(i_register) + list(anc_register))
        
        # Apply Grover operator
        qc.h(list(i_register))
        qc.x(list(i_register))
        qc.mcp(theta,list(i_register[:-1]), i_register[-1])
        qc.x(list(i_register))
        qc.h(list(i_register))
    
        qc.barrier()
\end{minted}
With this Grover step, the current quantum state is now $\frac{1}{\sqrt{\Omega}}\sum_{i=0}^{\Omega-1}|i\rangle\otimes|0\rangle$ (ignoring the ancilla qubit). As the ancilla is no longer required for any calculations, we will forget about its existence going forward.

Next, E needs to prepare the corresponding $|j_i\rangle$ for each $|i\rangle$ in the $j$-invariants register. This is done in two steps. First copy $|i\rangle$ into the $j$-register and then apply the ``fake'' mapper oracle $|i\rangle\mapsto|j_i\rangle$ to the $j$-invariants register. Using the correct mapper as described in Figures~\ref{fig:mapper} or~\ref{fig:RSq}, we would need to load $\ket{j_0}$ into the $j$-register and apply the correct mapper to both index- and $j$-register to create $|i\rangle\mapsto|j_i\rangle$. 

E also applies the shift with respect to $\omega$ to produce $|i\rangle\mapsto|j_{i+\omega}\rangle$.
\begin{minted}[mathescape,
  linenos,
  numbersep=5pt,
  gobble=2,
  frame=lines,
  framesep=2mm,
  breaklines,
  firstnumber=102]{python}
    # generate equal superposition of all |i,j_(i+omega)> states
    for j in range(len(i_register)):
        qc.cx(i_register[j],j_register[j])
    qc.unitary(UM,j_register,label="UM")
    for _ in range(omega):
        qc.unitary(UE,j_register,label="UE")
    
    qc.barrier()
\end{minted}
Thus, E has now created the state $\frac{1}{\sqrt{\Omega}}\sum_{i=0}^{\Omega-1}|i\rangle\otimes|j_{i+\omega}\rangle$. E should retain the index register and send the $j$-register to D who applies their shift.
\begin{minted}[mathescape,
  linenos,
  numbersep=5pt,
  gobble=2,
  frame=lines,
  framesep=2mm,
  breaklines,
  firstnumber=110]{python}
    # E sends j-register to D, D applies UD    
    qc.unitary(UD,j_register,label="UD")
\end{minted}
This creates the state $\frac{1}{\sqrt{\Omega}}\sum_{i=0}^{\Omega-1}|i\rangle\otimes|f_D(j_{i+\omega})\rangle$ where the left register is retained by E and D has the $j$-register. D further sends the $j$-register up the chain eventually reaching A via C and B who all execute their respective shift oracles.
\begin{minted}[mathescape,
  linenos,
  numbersep=5pt,
  gobble=2,
  frame=lines,
  framesep=2mm,
  breaklines,
  firstnumber=112]{python}
    # D sends j-register to C, C applies UC    
    qc.unitary(UC,j_register,label="UC")
    # C sends j-register to B, B applies UB
    qc.unitary(UB,j_register,label="UB")
    # B sends j-register to A, A applies UA    
    qc.unitary(UA,j_register,label="UA")
\end{minted}
At this point, the current quantum state is $\frac{1}{\sqrt{\Omega}}\sum_{i=0}^{\Omega-1}|i\rangle\otimes|f_A(f_B(f_C(f_D(j_{i+\omega}))))\rangle$ where the index register is retained by E and the $j$-register is held by A. A also attaches the encrypted message at this point, but we will ignore that for simplicity.

A now sends the registers back to E via B, C, and D who each use their inverse oracles on the $j$-register as well as their respective scrambling operations for the Diffie-Hellman layer on the message. 
\begin{minted}[mathescape,
  linenos,
  numbersep=5pt,
  gobble=2,
  frame=lines,
  framesep=2mm,
  breaklines,
  firstnumber=118]{python}
    # A sends state to B, B applies UBinv
    qc.unitary(UBinv,j_register,label="UBinv")
    # B sends state to C, C applies UCinv
    qc.unitary(UCinv,j_register,label="UCinv")
    # C sends state to D, D applies UDinv
    qc.unitary(UDinv,j_register,label="UDinv")
\end{minted}
Since the group action is commutative, the oracles commute; i.e., $f_{X^{-1}}\circ f_A\circ f_X=f_A\circ f_{X^{-1}}\circ f_X=f_A$ hold for $X\in\{B,C,D\}$, and we obtain that the final quantum state sent from D to E is $\frac{1}{\sqrt{\Omega}}\sum_{i=0}^{\Omega-1}|i\rangle\otimes|f_A(j_{i+\omega})\rangle$.

\begin{remark}
\begin{enumerate}
    \item Throughout this entire process, an attacker has only access to the $j$-register. Thus, measuring it yields exactly one of the states $|f_A(f_B(f_C(f_D(j_{i+\omega}))))\rangle$, $|f_A(f_C(f_D(j_{i+\omega})))\rangle$, $|f_A(f_D(j_{i+\omega}))\rangle$, or $|f_A(j_{i+\omega})\rangle$, depending on the point in the chain they apply their attack, with $i$ drawn uniformly. Thus, they will observe a random $j$-invariant drawn from a uniform distribution. 
    \item If the malicious attacker makes copies via imperfect cloning at a single point in the communication chain, then they will receive a sample of $j$-invariants as implemented in the cycle by E, possibly further shifted by generators from A, B, C, or D. While the attacker knows that these come from a $\Omega$-large sub-section of E's cycle, the order is unknown. Furthermore, due to imperfect cloning errors, some of these measurements will yield non-implemented states of the $j$-register, which may correspond to non-$j$-invariants or $j$-invariants that are not on the $\Omega$-size sub-cycle E has put into superposition. 
    \item If an attacker measures at multiple points in the chain, then the onion routing protocol has been compromised because the chain is supposed to be hidden. But even in that case, there are two options for the attacker. 
    \begin{enumerate}
        \item The attacker measures the state at each transmission. In that case, their attack is reduced to breaking the problem in the classical setting. 
        \item The attacker makes imperfect copies at each transmission. Then, they receive sets of $j$-invariants that but they do not know how they relate, there will be $\Omega!$ many possible neighbors within the data. Furthermore, there will be errors in the data that need to be handled. 
        For example, if the attacker measures between E and D as well as between D and C, then each imperfect copy may have a state infidelity of up to $1/6$. This means that with probability $5/6$ they measure a valid $j$-invariant from the implemented cycle, and with probability $1/6$ they end up in the erroneous contribution which comprises unknown contributions of invalid $j$-register states, $j$-invariants that do not correspond to the implemented cycle, and $j$-invariants that correspond to the implemented cycle. The attacker therefore needs to identify the false states and find the pairings $(j,\mathfrak{d}*j)$ of implemented $j$-invariants $j$ with their shifts $\mathfrak{d}*j$ as executed by D. There are $\Omega!$ many such pairing arrangements even if all false $j$-invariants have already been filtered out.
    \end{enumerate}
    \item Even in the case where the index register is not retained by E and the attacker intercepts at the last step, they obtain a random pair $(i,f_A(j_{i+\omega}))$ which cannot be used to compute $j_A=f_A(j_0)$ since $\omega$ is private information only E has. 
    
    The best case scenario an attacker can hope for is to observe the initial transmission from E to D and obtain the pair $(i,j_{i+\omega})$, as well as the final transmission from D to E where they obtain the pair $(i,f_A(j_{i+\omega}))$. Since the quantum state has collapsed to a classical state in the first observation, these will be the same value of $i$. Thus, a malicious attacker would need to find an element $\mathfrak{m} \in Cl(O_{\D})$ such that $\mathfrak{m}*j_0=j_{i+\omega}$, and use it to retrieve $j_A=\mathfrak{m}^{-1}*f_A(j_{i+\omega})$. Alternatively, the malicious attacker could intercept both upstream and downstream messages between any two actors attempt to find $\mathfrak{m}\in Cl(O_{\D})$ such that $\mathfrak{m}*f_X(j_{i+\omega})=f_A(f_X(j_{i+\omega}))$ for $f_X$ being the identity, $f_D$, $f_C\circ f_D$, or $f_B\circ f_C\circ f_D$ depending on the point of the attack, and use it to retrieve $j_A=\mathfrak{m}^{-1}*f_A(f_X(j_{i+\omega}))$. Solving any of these problems is, of course, as hard as breaking the scheme itself, as this is the central hardness assumption for schemes based on the ideal class group action. 
    
    Furthermore, even if this problem is solved and the malicious attacker gains access to $j_A$, then they still need to break the Diffie-Hellman encryption layer, i.e., they need to successfully execute two more such attacks.
    \item The transmission of the index register should be avoided, because using imperfect copies, an attacker could now get the additional information which $j$-invariants correspond to each other at different transmission points. This, of course, requires the chain of communication to be known, i.e., for the onion routing communication to already be partially compromised, but knowing the chain of communication now can lead to additional information that is not known in the classical the case and might compromise security. 
    \item If the index register needs to be contained in the transmission (e.g., for engineering reasons such as not being able to maintain entanglement if the index register is retained by E), then one should consider further encryption on the index- and $j$-registers. 
    \begin{enumerate}
        \item Adding a Diffie-Hellman encryption layer on the index- and $j$-registers prevents any external attacker from obtaining any information. However, an attacker measuring the state during transmission also destroys the state and E will only receive garbled non-sense. While this keeps the communication secure, this should be avoided as it opens up the possibility of an attacker preventing any communication reaching the receiver, unless the Diffie-Hellman layer is implemented in such a way that each classical index-$j$-combined register state is mapped to another classical index-$j$-combined register state. Furthermore, while this prevents any external attacker from executing the above-mentioned attack, each intermediary (B, C, or D) is still able to perform the attack; hence, must be trusted.
        \item E may add an encryption on the index register alone. Since nobody but E needs to access the index register, this does not impede other actors from executing their protocols. Again, an attacker measuring the index register will cause the communication to fail unless E's secret encryption maps classical states of the index register into classical states. In other words, the encryption must be chosen as a secret permutation of the numbers $0,\ldots,n-1$ with $\Omega\le n\le 2^{\lceil\log_2\Omega\rceil}$.
    \end{enumerate}
\end{enumerate}
    
\end{remark}

Upon receipt of this final transmission, all further operations are performed by E. First, E measures the state.
\begin{minted}[mathescape,
  linenos,
  numbersep=5pt,
  gobble=2,
  frame=lines,
  framesep=2mm,
  breaklines,
  firstnumber=124]{python}
    # E adds measurements
    qc.barrier()
    qc.measure(i_register,cl_i_register)
    qc.measure(anc_register,cl_anc_register)
    qc.measure(j_register,cl_j_register)
    
    # Transpile for simulator
    simulator = AerSimulator()
    qct = transpile(qc, simulator)
    
    # Run and get counts
    result = simulator.run(qct,shots=1).result()
    counts = result.get_counts(qct)
    key_raw = list(counts.keys())[0]
\end{minted}
This collapses the state into a single $|i\rangle\otimes|f_A(j_{i+\omega})\rangle$ which is uniformly selected from all $r$ possible values. As shown in Figure~\ref{fig:firstmeasurement}, repeating the experiment $10000$ times yields each state $|i\rangle\otimes|f_A(j_{i+\omega})\rangle$ roughly $\frac{10000}{\Omega}\approx2000$ times. 
\begin{figure}
    \centering
    \includegraphics[width=\linewidth]{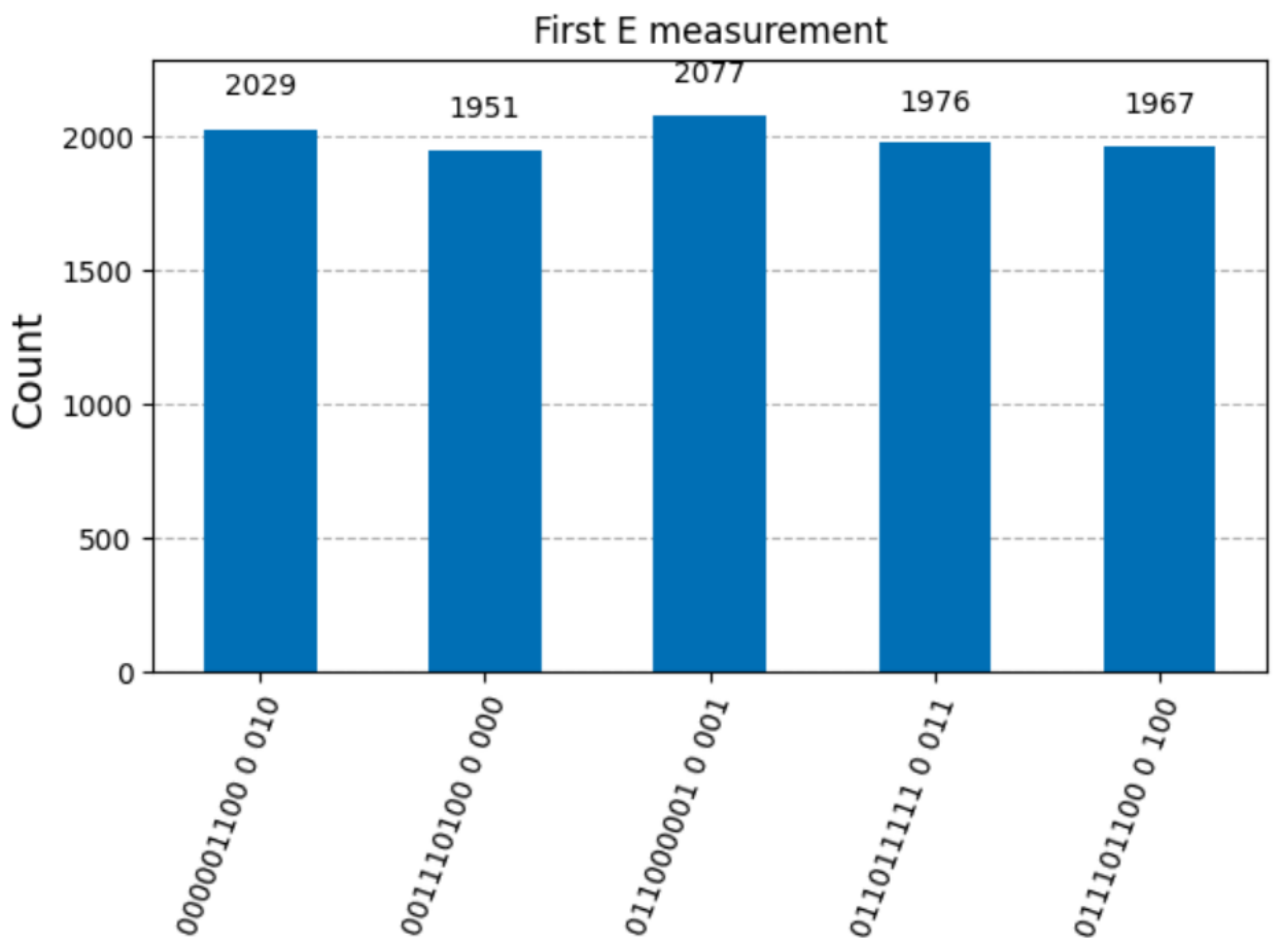}
    \caption{Example histogram of measurements as seen by E in their first measurement process if the communication were repeated 10000 times. The outcomes are drawn from a uniform distribution of all states $|i\rangle\otimes|f_A(j_{i+\omega})\rangle$. Here, the index register is the last set of three qubits, the $j$-register is the front set of nine qubits, and the single qubit in the middle is the ancilla required for the Grover step.}
    \label{fig:firstmeasurement}
\end{figure}

In reality, E will measure only once and obtain, for example, the measurement outcome \texttt{key\_raw = "011101100 0 100"}, which corresponds to $i=4$ and $f_A(j_{i+\omega})=236$. E now needs to uncompute their shift. They can do this by applying their inverse oracle \texttt{UEinv} $i+\omega$ times, or by applying the ``true'' inverse mapper oracle once and the inverse oracle \texttt{UEinv} $\omega$ times. Since we don't have the ``true'' mapper implemented, we will use \texttt{UEinv} $i+\omega$ times. E immediately measures the $j$-register afterwards again. 
\begin{minted}[mathescape,
  linenos,
  numbersep=5pt,
  gobble=2,
  frame=lines,
  framesep=2mm,
  breaklines,
  firstnumber=138]{python}
    # set up j-register as it is post measurement
    qc2 = QuantumCircuit(
        j_register,
        cl_j_register
    )
    
    key_split = key_raw.split(" ")
    for j in range(len(key_split[0])):
        if key_split[0][len(key_split[0])-j-1] == "1":
            qc2.x(j_register[j])
    
    qc2.barrier()

    # get index shift i 
    def bin2dec(s):
        n = 0
        for j in range(len(s)):
            n += int(s[len(s)-j-1])*2**j
        return n
    
    i = bin2dec(key_split[-1])
        
    # reverse shift i+omega
    for _ in range(i+omega):
        qc2.unitary(UEinv,j_register,label="UEinv")
        
    # and measure
    qc2.measure(j_register,cl_j_register)
\end{minted}
Letting E's cycle generator be denoted as $\mathfrak{e}$ and A's as $\mathfrak{a}$, this last operation executes
\begin{align*}
    \ket{\mathfrak{e}^{-i-\omega}*f_A(j_{i+\omega}}=\ket{\mathfrak{e}^{-i-\omega}*\mathfrak{a}*\mathfrak{e}^{i+\omega}*j_0}=\ket{\mathfrak{a}*j_0}=\ket{j_A}.
\end{align*}
Thus, the final measurement lets E read out $j_A$. 

In the example where \texttt{key\_raw = "011101100 0 100"}, E measures \texttt{"011010101"} in the $j$-register which is binary for $j_A=\mathfrak{a}*j_0=213$ as it should be since \texttt{aj[1]=213}. Hence, we have successfully communicated the secret key from A to E via B, C, and D.

\section{Security Considerations and Final Remarks}\label{sec:SecurityFinal}
For security reasons, we require the class number $h(\D)$ to be exponential in the security parameter. Since $h(\D)$ ``grows like'' $\sqrt{|\D|}$ \cite[pg.~19]{BhaMu}, $\D$, and hence the choice of $p$, should be set accordingly. Regarding the security of the class group action, there is ongoing work on appropriate parameter sets; the most up-to-date analysis is given by Campos et al.\ (2024) \cite{Campos}, which incorporates the best known quantum attacks based on Kuperberg’s hidden-shift algorithm running in subexponential time $2^{O(\sqrt{\log h(\D)})}$. As implementation techniques for evaluating the action continue to improve (e.g., \cite{PEGAS}), recommended parameters may be revised; a full treatment lies beyond the scope of this paper.

It is worth noting a distinction between a fully quantum setting such as ours and the classical group-action protocols (as used in post-quantum cryptography but analyzed against quantum adversaries). Classically, a public transcript may reveal a start point $j_0$ and an endpoint $j_k$, and the core hardness assumption is the vectorization problem: find $a\in G$ such that $a * j_0 = j_k$, for which the best known quantum algorithms are subexponential at best. In a quantum construction such as the QOR however, no intermediate $j$-invariants are accessible unless a measurement is performed (which yields a random $j$-invariant); only the public $j_0$ is visible. Moreover, by protocol design only the receiver performs the final measurement. 

Our second remark on security addresses the properties of \emph{uniform mixing} and \emph{periodicity} of the underlying isogeny graph. The property of uniform mixing is desirable since a sufficiently long walk on the isogeny graph becomes indistinguishable from a random walk. In particular, \emph{supersingular} $\ell$-isogeny graphs are \emph{Ramanujan}. Their nontrivial eigenvalues are bounded by $2\sqrt{\ell}$, so the classical walk mixes in $O(\log |V|)$ steps, $|V|$ being the number of vertices (see, e.g. \cite{CGL}). For \emph{ordinary} isogeny graphs Jao–Miller–Venkatesan (under GRH) prove in \cite[Cor.~1.3]{JaoExpander} that a walk of length
\[
t \;\ge\; C\,\frac{\log |V|}{\log\log q}
\]
is already near-uniform from any start. Here $C$ is a positive constant and, in the elliptic curve scenario of \cite[Theorem 1.5]{JaoExpander}, $|V| \sim \sqrt{q}$. 

In the quantum setting on the other hand, uniform mixing is known only for a narrow family of graphs, so one can instead work with the relaxed notion of $\varepsilon$-uniform mixing. It is known that for every $p \geq 5$, every $p$-cycle $C_p$ exhibits $\varepsilon$-uniform mixing \cite[Theorem 5.5.2]{Mull}. Consequently, for appropriately chosen walk lengths (or evolution times), the walkers' paths appear \(\varepsilon\)-random to any observer. Moreover, combining \cite[Theorem 2.2 \& Corollary 2.3]{Godsil2} and \cite[Theorem 4.3.2]{Mull}, one obtains the known result that no $n$-cycle $C_n$ is periodic unless $n \in \{1, 2, 3, 4, 6\}$. Finally, \cite[Lemma 3]{SaSeShpa} ensures that other than the complete graph $K_p$, no other circulant graph $G(p;S)$ on $p$-many vertices and generating set $S$ of cardinality $|S| < p-1$ is periodic. It therefore appears that $p$-cycles are good candidates for quantum cryptographic scenarios as the one considered in this paper. In the extremal case where the class group is cyclic of large prime order, the resulting isogeny graph, after $d$-many encryptions, is the disjoint union of $d$-many $p$-cyles $C_p$. Even when the class group is not of prime order, one can restrict to a large enough prime-order subgroup and the corresponding isogeny cycles to obtain the same effect.

Our final remark turns to implementation considerations. These primarily concern the quantum oracles. If we are in the CSIDH/CSURF setting discussed in Section~\ref{sec:ClAction}, i.e. we work with supersingular elliptic curves over $\mathbb{F}_p$, then efficient classical computation of the action $*$ by arbitrary class group elements is becoming increasingly feasible, owing to advances such as PEGASIS~\cite{PEGAS}. This remains an active research area, and we do not pursue further details here as this is not in the scope of the paper. Alternatively, as already discussed in the paper, the class group action can be interpreted as continuous-time quantum walk on the underlying isogeny graph, hence on the underlying Hamiltonian. Currently, in a parallel project, we are developing this viewpoint theoretically, and we are considering possible implementations on actual hardware.

\section*{Acknowledgments}
This work is supported by the European Union’s Horizon Europe Framework Programme (HORIZON) under the ERA Chair scheme “QUEST” with grant agreement no. 101087126. This work is also supported with funds from the Ministry of Science, Research and Culture of the State of Brandenburg within the Centre for Quantum Technologies and Applications (CQTA).

\begin{figure}[h]
   \includegraphics[width = 0.12\textwidth]{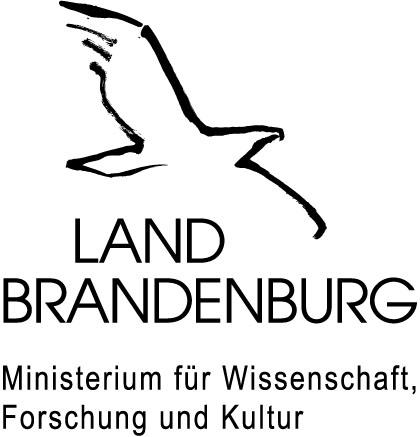}
\end{figure}

\end{document}